\documentclass[12pt]{article}
\usepackage{amsmath}
\usepackage{amsthm}
\usepackage{amsfonts}
\usepackage{mathtools}
\usepackage[english]{babel}
\usepackage{hyperref}
\usepackage{authblk}
\allowdisplaybreaks
\usepackage[a4paper, margin=1in]{geometry}
\usepackage{setspace}
\usepackage{mathtools}
\usepackage{graphicx}
\usepackage{epstopdf}
\usepackage{subcaption}
\usepackage{float}
\usepackage{dsfont}

\newcommand\PS{\mathbb{Q}_{\mathrm{S}}}

\newcommand{\DF}{\mathrm{DF}}
\newcommand{\UF}{\mathrm{UF}}
\newcommand{\PV}{\mathrm{PV}}
\newcommand{\bid}{\mathrm{bid}}
\newcommand{\ask}{\mathrm{ask}}
\newcommand{\bidarg}{\mathrm{bid}(\lambda, \gamma)}
\newcommand{\askarg}{\mathrm{ask}(\lambda, \gamma)}
\newcommand{\bidargprot}{\mathrm{bid}^{\text{prot}}(\lambda, \gamma)}
\newcommand{\askargprem}{\mathrm{ask}^{\text{prem}}(\lambda, \gamma)}
\newcommand{\bidargprem}{\mathrm{bid}^{\text{prem}}(\lambda, \gamma)}
\newcommand{\askargprot}{\mathrm{ask}^{\text{prot}}(\lambda, \gamma)}
\newcommand{\lgd}{\mathrm{LGD}}
\newcommand{\overbar}[1]{\mkern 1.5mu\overline{\mkern-1.5mu#1\mkern-1.5mu}\mkern 1.5mu}

\newtheorem{Assumption}{Assumption}
\newtheorem{lemma}{Lemma}
\newtheorem{theorem}{Theorem}

\date{14 April 2021}

\begin{document}


\title{From bid-ask credit default swap quotes to risk-neutral default probabilities using distorted expectations}

\author[,1,2]{Matteo Michielon\thanks{matteo.michielon@nl.abnamro.com (corresponding author).}}
\author[,2]{Asma Khedher\thanks{a.khedher@uva.nl.}}
\author[,2,3]{Peter Spreij\thanks{p.j.c.spreij@uva.nl.}}
\affil[1]{\small Quantitative Analysis and Quantitative Development, ABN AMRO Bank N.V., Gustav Mahlerlaan 10, 1082 PP Amsterdam, The Netherlands.}
\affil[2]{\small Korteweg-de Vries Institute for Mathematics, University of Amsterdam, Science Park 105-107, 1098 XG Amsterdam, The Netherlands}
\affil[3]{\small Institute for Mathematics, Astrophysics and Particle Physics, Radboud University Nijmegen, Huygens building, Heyendaalseweg 135, 6525 AJ Nijmegen, The Netherlands}

\maketitle

\begin{abstract}
Risk-neutral default probabilities can be implied from credit default swap (CDS) market quotes. In practice, mid CDS quotes are used as inputs, as their risk-neutral counterparts are not observable. We show how to imply risk-neutral default probabilities from bid and ask quotes directly by means of formulating the CDS calibration problem to bid and ask market quotes within the conic finance framework. Assuming the risk-neutral distribution of the default time to be driven by a Poisson process we prove, under mild liquidity-related assumptions, that the calibration problem admits a unique solution that also allows to jointly calculate the implied liquidity of the market.
\end{abstract}

\begin{doublespace}

\section{Introduction}
Risk-neutral default probabilities play a crucial role in modeling (counterparty) credit risk, as for instance in valuation adjustment calculations, and an approach that is often followed is that of computing them starting from credit default swap (CDS) market quotes. This article aims to relate risk-neutral default probabilities and CDS quotes in a two-price economy within the conic finance paradigm by means of providing a methodology to extract the former from bid and ask quotes directly, i.e., without relying on any mid price approximation. Before explaining our contribution in detail, we first provide a review of the relevant literature concerning bid-ask pricing and conic finance.

Bid-ask pricing can be modeled, in a consistent manner with risk-neutral valuation, in different manners. A possible way to do so is that of transforming the risk-neutral measure via appropriate concave distortion functions as per \cite{cherny2009}. This approach, known as conic finance and introduced in \cite{cherny}, is based on the idea of modeling illiquid markets as abstract entities accepting, at zero cost, a convex cone of random variables containing the non-negative cashflows. By balancing risks and rewards to assess the ``quality'' or ``expected performance'' of contingent claims via the concept of \emph{index of acceptability}, this framework allows to use Choquet expectations \cite{choquet1953} as building blocks for computing bid and ask prices. 

The former conic modeling framework, which has triggered extensive research and of which several applications are available in \cite{bookConicFinance}, employs a static notion of index of acceptability, which allows to choose amongst cashflows, at the valuation date, based on their (cumulative) expected terminal value. This idea has been further extended by \cite{dynamicAccIndex} and later by \cite{biaginiNadal}, amongst others, to a dynamic setup, where  \emph{dynamic indices of acceptability} are defined in a multi-period setting. Dynamic acceptability indices allow to re-assess the initial classification of the traded cashflows on the basis of the latest information available, consistently over time, in the sense that future preferences are conforming with the current ones. This has lead to the possibility of pricing and hedging in a \emph{dynamic conic finance} framework for finite probability spaces in a discrete-time setting  as in \cite{dynamicConicFinance}, where time-dependent bid and ask prices of contingent claims, potentially including dividends and transaction costs, are expressed in terms of dynamic indices of acceptability, and where a (dynamically-consistent) version of the First Fundamental Theorem of Asset Pricing is provided in terms of no-good-deal conditions. For a unified framework for the time-consistency between dynamic risk measures and dynamic performance measures in discrete time, see \cite{unified}, while a survey concerning the time-consistency property of dynamic risk and performance measures is available in \cite{survey}.

Security prices do not only depend on the direction of the transaction but also on the size of the order, and different approaches can be considered to include this additional liquidity charge in the relevant pricing equations. \cite{nadal} introduces, in continuous time, an approach providing dynamic bid and ask processes for contingent claims which include both the aforementioned liquidity effect, as well as transaction costs. By replacing scale-invariance with sub-scale-invariance \cite{rosazza} develop a dynamic framework, in a continuous-time setup and given a general probability space, that also captures this additional liquidity cost and that allows to value financial securities in terms of \emph{g-expectations} \cite{peng} (comparisons between the definitions of $g$- and Choquet expectations are available in \cite{comparison} and \cite{comparison2}). Again on general probability spaces, both liquidity and transaction costs can be included within a dynamic conic finance approach where pricing is based on $g$-expectations as in \cite{dynamicConicFinance2}.

Within the conic finance paradigm different studies deal with credit-related topics. \cite{eberlein2012} show that, if assets and liabilities are marked at the bid and at the ask, respectively, then the potential accounting profitability of a firm induced by its own credit quality deterioration is eliminated. These ideas are further applied in \cite{cva1} in the case of credit and debit valuation adjustments. \cite{madanCds} proposes an approach to estimate the parameters of risk acceptability of CDSs and their time dependence, and applies the methodology to a period including, but not limited to, the 2008 financial crisis. Therein, the industry practice of taking mid CDS quotes to proxy their risk-neutral counterparts is adopted and, for each CDS, a flat hazard rate term structure is considered in the calibration. Further, within the dynamic conic finance framework, bid and ask price processes for CDSs are constructed in \cite{dynamicConicFinance} and in \cite{dynamicConicFinance2}. 

A methodology that allows to jointly calibrate a CDS model to bid and ask market quotes and to imply risk-neutral default probabilities without computing them from mid quotes is not yet available. In the present article we provide an approach to tackle this problem within the conic finance paradigm. The economic rationale behind our research question is given by the fact that, while in practice model parameters are usually calibrated starting from mid quotes as proxies for their risk-neutral counterparts, in reality a security trades neither at the risk-neutral nor at the mid price, but instead either at the bid or at the ask, depending on the direction of the trade. Thus, one might then want to be able to include in a simple manner the liquidity effect within their CDS pricing equations, for instance as CDS markets for single name CDSs not being amongst the most liquid; see \cite{cdsLiquidity}. Given our aim of extracting risk-neutral probabilities from the currently-observed bid and ask CDS quotes, a static approach to conic finance suffices. Moreover, this allows to easily define a term structure for the liquidity level of the CDS market, and also to restate the bid-ask calibration problem in terms of recursively solving a non-linear constrained system. In the case of CDSs, modeling the default time via a reduced-form model by explicitly specifying the functional form of its distribution is a popular choice. In particular, we consider the case of a Poisson process driving the dynamics of the credit event, of which the standard International Swaps and Derivatives Association (ISDA) model \cite{openGamma}, which is a common choice  amongst financial practitioners, is a possible specification. In these settings, we show that the bid-ask CDS calibration process has, under some mild assumptions, a unique solution. Further, the methodology proposed here allows to jointly strip \emph{implied liquidity} parameters for CDS markets in the spirit of \cite{impliedLiquidity} with a term structure. To the best of our knowledge, this is the first attempt to calibrate a credit model using Choquet expectations to bid and ask CDSs quotes directly without relying on approximating their risk-neutral counterparts with the respective mid quotes and, thus, our contribution is novel.

\medskip

This paper is organized as follows. In Section \ref{sec:valuationOfCdss} we recall the basics of CDS valuation in the risk-neutral framework, how Poisson processes can be used to model the default time for CDS valuation purposes, as well as their calibration to market data. In Section \ref{sec:bidAsk}  we briefly recall how pricing via distorted expectations works, while in Section \ref{sec:twoPriceEconomy}, we introduce the CDS bid-ask calibration problem in the settings of \cite{cherny}. We show that the problem admits, under simple assumptions, a unique solution. We also provide a calibration example, based on the standard ISDA model, which is a special case of Poisson-based CDS model. Section \ref{sec:conclusion} concludes.

\section{Basic notions and valuation of CDSs}\label{sec:valuationOfCdss}
A CDS is a bilateral derivative contract which involves the transfer of the credit risk arising from bonds or other forms of debt issued, amongst others, by corporates, municipalities, or sovereign states. Thus, a CDS is a sort of insurance policy, as it provides the protection buyer, who might or might not own the underlying credit\footnote{As CDSs do not require the buyer of the contract to hold the insured asset.}, with protection against a \emph{credit event}. The formal definition of a credit event is contract-specific and complex from the legal angle. Therefore, from here onward, the expressions \emph{credit event} and \emph{default} will be used to refer to a set of circumstances that trigger the protection payment. 

A CDS contract involves two parties, i.e., a \emph{protection buyer} and a \emph{protection seller}. The protection seller commits to compensate for the (potential) loss of the counterparty if a default event for the reference entity occurs within a predetermined time frame. A CDS can be therefore seen as a derivative contract where the underlying is the default time of the issuing entity. CDSs were initially mainly physically settled: if default event occurs, then the protection buyer delivers one of the defaulted bonds of the reference entity to the protection seller, in return for its par value. However, due to the size of the CDS market it might happen that, should there not be enough supply of defaulted bonds in the market, an auction is conducted to determine what the recovery value of the defaulted bond is. In this case the CDS contract is, thus, cash settled, and this is nowadays the most common settlement practice (for further details refer to \cite{CDSISDADef2003}).  The standardization process of credit derivatives led by the ISDA,  see \cite{ISDAOTC, CDSISDADef2003,CDSISDADef2014}, has introduced conventions on the way these contracts are traded. These conventions can be region specific: for example, some conventions for North-American CDSs (\emph{CDS big bang}; see \cite{CDSBigBang}) might differ from those of European CDSs (\emph{CDS small bang}; see \cite{MarkitSmallBang}). Before the CDS standardization process,  in a similar fashion to interest rate swaps, CDSs used to be quoted at par, i.e., the coupon rate was defined such that the contract had zero value at inception, for both parties. However, CDSs have now standard coupons and, as a consequence, a non-zero entry cost called \emph{upfront payment}, which is payed on the cash settlement date and that is usually quoted as a percentage (i.e., as \emph{points upfront}) of the notional amount. The upfront payment can be interpreted as an amount reflecting the difference in value between a par CDS and one with a given standard coupon.

CDSs can be used to estimate the default probabilities of a wide range of issuing entities by means of appropriate pricing models calibrated using the available CDS market data. These implied probabilities can be used, for instance, as inputs in various valuation adjustment calculations,  which makes CDS useful for hedging purposes; see \cite[Ch. 4 and Ch. 12]{green}.

We introduce now the essential notations and conventions that define CDSs; refer, for instance, to \cite{openGamma} for a detailed overview. In a standard CDS contract the \emph{CDS dates} are the semi-annual termination dates of the CDS, and fall on 20 March and 20 December of each year.\footnote{Before 20 December 2015 the frequency of the CDS roll dates was quarterly instead of semi-annual, with resulting termination dates falling on 20 March, 20 June, 20 September and 20 December of each year; see \cite{ISDAchange}.} From the \emph{protection effective date} (i.e., $t_p$) the protection starts; this date is generally defined as the valuation date plus one day. The \emph{cash settlement date} (i.e., $t_s$) is when any upfront payments are made, and can be lagged by a few business days compared to the valuation date (the standard ISDA model defines it as the valuation date plus three business days; see \cite{openGamma}).  The \emph{accrual start dates} (i.e., $s_1, \dots,s_N$) are used as starting points for calculating the coupon payments. This increasing sequence contains all the CDS dates before the maturity date, with $s_1$ set as the previous CDS date before the protection effective date. This is because holding a CDS over a coupon payment entails paying or receiving the full coupon payment amount. The \emph{accrual end dates} (i.e., $e_1,\dots,e_{N}$) are the dates used as end points for calculating the premium payments, with $e_N$ the maturity of the contract. Premium payments are made by the protection seller to the protection buyer at the \emph{payment dates} (i.e., $t_1,\dots,t_N$). We denote with $\lgd$ the loss-given-default expressed per unit of notional (i.e., one minus the recovery rate), that we assume to be constant, and with $\Delta(t,s)$ the year fraction between $t$ and $s$ ($t<s$); see \cite{openGamma} for further details concerning day-count conventions. In particular, we use the shorthand notation $\Delta_i$ instead of $\Delta(s_i,e_i)$. Further, $N$ denotes the notional amount, $\tau$ the default time, and $\mathds{1}_{\left\{\,\cdot\,\right\}}$ the indicator function.

The protection leg is the contingent payment the protection seller makes to the protection buyer. Despite in practice there is usually a lag between the default time and the protection payment, modeling-wise at $\tau$ the protection seller is assumed to pay to the counterparty the amount
\begin{equation}
\lgd\cdot N\cdot\mathds{1}_{\left\{t_p\leq\tau\leq e_N\right\}}.
 \end{equation}

The premium leg is defined as the series of payments the protection buyer makes to the counterparty until either a credit event occurs or the contract expires. We denote its fixed coupon, per unit of notional, with $C$. The amount paid by the protection buyer to the protection seller on each payment date $t_i$ is given by
\begin{equation}\label{protectionCoupon}
C\cdot N \cdot\Delta_i \cdot\mathds{1}_{\left\{\tau> e_i\right\}}.
\end{equation}

In the case of a credit event, the protection buyer pays to the counterparty the \emph{accrued coupon}, i.e., if $\tau\in[s_i,e_i]$, the accrued coupon payed upon default equals
\begin{equation}\label{accruedCoupon}
C\cdot N \cdot\Delta(s_i,\tau)\cdot \mathds{1}_{\left\{s_i\leq\tau\leq e_i\right\}}.
\end{equation}

On a filtered probability space $(\Omega, \mathcal{F}, (\mathcal{F}_{t})_{t\in[0,T]}, \mathbb{P})$, with $\mathbb{P}$ the real-word probability measure, we denote with $\mathbb{Q}$ a risk-neutral measure, with $\mathbb{E}^\mathbb{Q}(\,\cdot\,)$ the expectation, at valuation date, with respect to $\mathbb{Q}$, with $\DF(t)$ the discount factor from $t$ to valuation date, while with $\PS(t)$ the survival probability of the reference entity until time $t$, i.e., $\PS(t)\coloneqq\mathbb{Q}(\tau>t)$. From here onward we assume, without loss of generality, unit notionals.

For the protection buyer the value of a CDS equals the value of its protection leg minus the one of its premium leg. In symbols
\begin{equation}\label{eq:pvProt}
	\mathrm{PV}^{\text{prot}} =  \lgd\cdot\mathbb{E}^\mathbb{Q}\left(\DF(\tau)\cdot\mathds{1}_{\left\{t_p\leq \tau\leq e_N\right\}}\right),
\end{equation}
while the value of the premium leg is given by
\begin{equation}\label{eq:pvPrem}
	\mathrm{PV}^{\text{prem}} = C\sum_{i=1}^N \mathbb{E}^\mathbb{Q}\left(\DF(t_i)\cdot \Delta_i \cdot  \mathds{1}_{\left\{\tau>e_i\right\}}+\DF(\tau) \cdot\Delta(s_i, \tau) \cdot  \mathds{1}_{\left\{s_i\leq \tau\leq e_i\right\}}\right).
\end{equation}
The present value of the CDS, from the perspective of the protection buyer, is defined as
\begin{equation}\label{eq:pv}
	\mathrm{PV}^{\text{CDS}} \coloneqq \mathrm{PV}^{\text{prot}} - \mathrm{PV}^{\text{prem}}.
\end{equation}

\subsection{Dynamics of the survival probabilities}\label{dynamicsOfTheSurvivalProbabilities}

There are different approaches to model the dynamics of the default time of an issuing entity, and the category of \emph{reduced-form} (or \emph{intensity}) models is one of these. In reduced-form models the probability distribution of the credit event is modeled directly; two well-known illustrations of models belonging to this class are, amongst others, \cite{JarrowAndTurnbull}, where a discrete Poisson bankruptcy process is presented, and \cite{DuffieAndSingleton}, where the risk-free discounting short-rate process is augmented with an instantaneous intensity process to account for credit risk. Reduced-form models are fundamentally different, for instance, from \emph{structural} (or \emph{firm-value}) models, which characterize defaults as consequences of events such as the value of a firm being too low for covering its liabilities, of which the so called Merton's 1974 firm-value model \cite{merton74} is an illustration. This idea has been later extended in \cite{black76}, where default occurs when the value of the firm's asset falls below a given threshold level, and that is considered the first prototype of the so called \emph{first passage time} models. While the main advantage of structural models is that of their consistency with the capital structure of the firm, they require firm-specific information that is not necessarily easily available. Thus, the main difference between the reduced-form and the structural approaches is given by the fact that default is something exogenous in the former, while endogenous in the latter. The idea of modeling the default probability distribution directly as done in reduced-form models allows, at least in theory, to simplify the problem tractability, as modeling the default event per se is easier than modeling the economic situations that might cause it. This often makes reduced-form models preferable to structural ones for practical applications, as done in this article, given that we explicitly consider the CDS market as source of information. For more details concerning different approaches to credit risk modeling the reader can refer, amongst others, to \cite{rutkowski}.

In the context of reduced-form models a possible approach to introduce a term structure for the distribution of the default time is that of defining the survival probability $\PS$ via
\begin{equation}\label{PSlambda}
	\PS(t)\coloneqq e^{-\int_0^t \lambda(s)\,ds}, 
\end{equation}    
\noindent where the deterministic function $\lambda:[0,+\infty)\rightarrow(0,+\infty)$ is called \emph{hazard rate} (or \emph{default intensity}) function.

Assuming that $K$ CDS quotes for a given reference entity with the same fixed coupon are available in the market, maturing respectively at $e_{N_1},\dots,e_{N_K}$ where $e_{N_1}<\ldots<e_{N_K}$, a possible way to define the hazard rate function with a term structure is given by setting 
\begin{equation}\label{lambda}
\lambda(t)\coloneqq
\begin{cases}
l_1(\lambda_1;t) &\mbox{if } 0\leq t\leq e_{N_1}\\
l_2(\lambda_1,\lambda_2;t) &\mbox{if } e_{N_1}< t\leq e_{N_2}\\
\ \vdots&\\
l_K(\lambda_1,\ldots,\lambda_K;t) &\mbox{if } e_{N_{K-1}}< t
\end{cases}\!\!,
\end{equation}      
\noindent where $\lambda_i>0$, $l_i(\lambda_1,\ldots,\lambda_i;\,\cdot\,)$ is deterministic and continuous, and with the function $l_i(\lambda_1,\ldots,\lambda_{j-1},\,\cdot\,,\lambda_{j+1},\ldots,\lambda_i;t)$ increasing, for $1\leq j\leq i$. The parameters $\lambda_1,\ldots,\lambda_K$ are those that, once set, specify the distribution of the default time. Common specifications for (\ref{lambda}) are, among others, piecewise-constant and piecewise-linear.\footnote{That is, given $K$ positive values $\bar{\lambda}_1,\ldots,\bar{\lambda}_K$, in the piecewise-constant case, for every $i$, $l_i(\lambda_1,\ldots,\lambda_i;t)\equiv \bar{\lambda}_i$. For the piecewise-linear case, on the other hand, we have $l_1(\lambda_1;t)\equiv \bar{\lambda}_1$, $l_i(\lambda_1,\ldots,\lambda_i;t)=\bar{\lambda}_{i-1}+\frac{\bar{\lambda}_{i}-\bar{\lambda}_{i-1}}{e_{N_i}-e_{N_{i-1}}}\cdot(t-e_{N_{i-1}})$ for $2\leq i\leq K-1$, while $l_K(\lambda_1,\ldots,\lambda_K;t)\equiv \bar{\lambda}_K$.} The first option provides the simplest assumption possible concerning the behavior of the hazard rate function across CDS maturities, and it as well results in better numerical stability compared to its piecewise-linear counterpart; further, modeling the default time via this simple approach is often enough for practical applications such as for its usage in several credit valuation adjustment calculations; see \cite[Sec. 4.4]{green}. Moreover, note that \eqref{eq:pv} is a model-independent relationship which assumes interest rates being independent from the default time. In this case, see \cite{Brigo}, CDS models can be calibrated to match CDS quotes exactly, and the resulting implied default probabilities calculated using different model specifications are expected to be in line with each other. In \cite{Brigo} this fact is illustrated by taking into account Lehman Brothers CDSs during different periods between August 2007 until the bank files for bankruptcy in September 2008. In particular, therein a comparison between the default probabilities implied using the Analytically-Tractable First Passage (AT1P) model and the intensity model with hazard rate function defined as per \eqref{lambda} in a piecewise-constant manner is provided (the AT1P model is a first-passage time structural model where default events are triggered by a firms' assets value hitting a deterministic threshold). The results show that the two models, despite their differences in terms of specifications, produce extremely close default probabilities (i.e., the largest difference observed at the calibration maturities is of the order of 0.8\%). Hence, as we are interested in implying default probabilities at the valuation time, this further justifies the choice of the modeling approach we have followed: if little model risk is linked to the model specifications used to extract the default probabilities, then model simplicity and tractability should be encouraged.

To calibrate the model parameters, we denote with $\UF^{\text{bid}}_i$ ($\UF^{\text{ask}}_i$) the bid (ask) upfront premium of the $i$\textsuperscript{th} quoted CDS contract. Their mid counterparts are denoted as $\UF^{\text{mid}}_i$. The values $\lambda_1,\ldots,\lambda_K$ are computed to match the quoted CDS market values. Risk-neutral premia are not observable, and they are usually proxied with their mid counterparts.

Due to quoting convention, the first upfront premium is defined such that the equality
\begin{equation}\label{eq:lambda1}
	\PV^{\text{CDS}}_1(\lambda_1) + \mathrm{Acc} =\DF(t_s)\cdot \UF_1^{\text{mid}}
\end{equation}
holds, where $\PV^{\text{CDS}}_1(\lambda_1)$ denotes the present value of the first CDS, as a function of $\lambda_1$, and where $\mathrm{Acc}$ equals $\DF(t_s)\cdot C\cdot\Delta(s_1, t_p)$. We can solve for $\lambda_1>0$ such that \eqref{eq:lambda1} is satisfied.\footnote{The higher the values reached by the hazard rate function, the higher the chances are that there will be a default. Thus, the more the protection seller wants to be paid to sell insurance. One would then intuitively expect $\PV^{\text{CDS}}_1(\lambda_1),\ldots,\PV^{\text{CDS}}_K(\lambda_K)$ to be strictly increasing in $\lambda_1,\ldots,\lambda_K$, respectively.  In \ref{sec:remark} we show that, for common coupon and $\lgd$ values, the value of the $i$\textsuperscript{th} CDS calculated using the setup outlined in the current section is strictly increasing in $\lambda_i$ when $i>1$, and that the same holds when $i=1$, at least when $\lambda_1$ belongs to an interval wide enough for practical purposes. The strict monotonicity of $\PV^{\text{CDS}}_i(\lambda_i)$  guarantees that, if $\mathrm{Acc} -\DF(t_s)\cdot \UF_i^{\text{mid}}\in\PV^{\text{CDS}}_i([0,+\infty))$, the equation $\PV^{\text{CDS}}_i(\lambda_i) + \mathrm{Acc} =\DF(t_s)\cdot \UF_i^{\text{mid}}$ admits a unique solution and, as a consequence, that the CDS calibration problem is well-defined. From here onwards we will always assume this to be the case.}

Then, we can consider the second upfront premium. By using the value of $\lambda_1$ computed above, one can imply $\lambda_2>0$ such that
\begin{equation}
	\PV^{\text{CDS}}_2(\lambda_2) +\mathrm{Acc} = \DF(t_s)\cdot\UF_2^{\text{mid}},
\end{equation}
with $\PV^{\text{CDS}}_2(\lambda_2)$ the present value of the second CDS as a function of $\lambda_2$.

By proceeding inductively for the remaining indices this procedure allows to define a term structure for the default probabilities that is in line, via (\ref{lambda}), with the mid quotes ``observed'' in the market.

Note that no choice of the hazard rate function in (\ref{PSlambda}) comes without problems. For instance, the simple possible choice of assuming piecewise-constant hazard rates, which is at the base of the so called standard ISDA model, can produce negative hazard rates under specific market circumstances; see \cite[Ch. 4.4.3]{green}. Therefore, depending on the market conditions, some functional forms for the hazard rates can be more suitable than others. It is thus necessary to assume, for the chosen functional form of the hazard rate function, that risk-neutral quoted values allow the model to be properly specified, as well as the calibration problem under one-price settings to be successful.

\section{Bid-ask pricing via distorted expectations}\label{sec:bidAsk}

An index of acceptability is a map $\alpha: L^\infty(\Omega, \mathcal{F}, \mathbb{P}) \to [0,+\infty]$ aiming to measure the quality of random cashflows, i.e., for a given contingent claim $X$ the higher the value of $\alpha(X)$, the higher $X$ is ranked. We say that $X$ is acceptable by the market at level $\gamma$ whenever $\alpha(X)\geq\gamma$. An index of acceptability $\alpha$ is expected to satisfy some basic properties. Namely, if both $X$ and $X'$ are acceptable at level $\gamma$, then also $\lambda\cdot X+(1-\lambda)\cdot X'$ for $\lambda\in[0,1]$ is (quasi-concavity property). $\alpha$ is assumed to be monotonic, i.e., if $X\geq X'$ then $\alpha(X)\geq\alpha(X')$, as well as scale-invariant, i.e., $\alpha(\lambda\cdot X)=\alpha(X)$ for every $\lambda>0$. Lastly, $\alpha$ is assumed to satisfy the Fatou property, which means that, if $(X_n)_n$ is a sequence of random variables such that, for every $n$, $|X_n|\leq 1$ and $\alpha(X_n)\geq \gamma$, then if $(X_n)_n$ converges in probability to a random variable $X$, also $\alpha(X)\geq \gamma$. It can be proven, see \cite{cherny2009}, that given an index of acceptability $\alpha$, for every $x\geq0$ there exists a set $\mathfrak{Q}_x$ of probability measures absolutely continuous with respect to $\mathbb{P}$ such that 
\begin{equation}
\alpha(X)=\sup\left\{x\geq0:\inf_{\mathbb{Q}\in\mathfrak{Q}_x}\mathbb{E}^{\mathbb{Q}}(X)\geq0\right\}
\end{equation}
and, further, if $x\leq x'$ then $\mathfrak{Q}_x\subseteq\mathfrak{Q}_{x'}$.

A \emph{coherent risk measure} is a functional $\rho:L^\infty(\Omega, \mathcal{F}, \mathbb{P}) \to [0,+\infty]$ that satisfies the transitivity, sub-additivity, positively homogeneity and monotonicity properties; see \cite[Ch. 4.1]{bookConicFinance}.\footnote{$\rho$ is said to be transitive (or translation-invariant) when $\rho(X+\lambda)=\rho(X)+\lambda$ for every $\lambda\in\mathbb{R}$, sub-additive when $\rho(X+X')\leq\rho(X)+\rho(X')$, positively homogeneus when $\rho(\lambda\cdot X)=\lambda\cdot\rho(X)$ for every $\lambda>0$, and monotonic when $\rho(X)\leq\rho(X')$ if $X\leq X'$. Note that the definition of coherent risk measure introduced in \cite{artzner} differs from the one provided here in the sense that, in \cite{artzner}, cash-invariance reads $\rho(X+\lambda)=\rho(X)-\lambda$, where $\lambda\in\mathbb{R}$, while monotonicity as $\rho(X)\geq\rho(X')$ when $X\leq X'$ (refer to \cite[Sec. 4.2.1]{grabisch} for some remarks concerning these differences). Given that we consider here coherent risk measures within the conic finance paradigm, we adopt therefore the definition outlined in \cite{bookConicFinance}.} It can be shown, see \cite{delbaen2009}, that a coherent risk measure can be identified with a functional of the form $\sup_{\mathbb{Q}\in\mathfrak{Q}}\mathbb{E}^{\mathbb{Q}}(X)$, where $\mathfrak{Q}$ is a set of probability measures absolutely continuous with respect to $\mathbb{P}$. Therefore, the level of acceptability of a cashflow $X$ can be rewritten in terms of coherent risk measures, i.e., as
\begin{equation}\label{eq:ia}
\alpha(X)=\sup\left\{ x\geq0:\rho_x(-X)\leq0 \right\},
\end{equation} 
where $(\rho_x)_{x\geq0}$ is a family of coherent risk measures such that $\rho_x(-X)\leq\rho_{x'}(-X)$ whenever $x\leq x'$. From this, it then follows that $\alpha(X)\geq\gamma$ if and only if  $\rho_\gamma(-X)\leq0$.\footnote{If $\rho_\gamma(-X)\leq0$, by \eqref{eq:ia} it follows that $\alpha(X)\geq\gamma$. On the other hand, assume that $\rho_{\gamma}(-X)>0$. Then, $\rho_x(-X)\geq\rho_{\gamma}(-X)>0$ when $x\geq\gamma$, from which $\alpha(X)<\gamma$, contradiction.\label{fn:condition}} Note that, for every $x\geq0$, one can define the \emph{acceptability set} associated with $\alpha$ as $\mathcal{A}_x\coloneqq\left\{ X\in L^\infty(\Omega, \mathcal{F}, \mathbb{P}):\rho_x(-X)\leq0 \right\}$. It then follows that $\left( \mathcal{A}_x \right)_{x\geq0}$ is a family of convex cones, each containing the non-negative random variables, with size decreasing in $x$. Thus, given an index of acceptability and a family of coherent risk measures $\left( \rho_x\right)_{x\geq0}$, for every acceptability level $x$ we obtain a convex cone $\mathcal{A}_x$ of contingent claims that are acceptable for the market, from which the term \emph{conic finance} originates.

The (asymmetric) Choquet integral of $X$ with respect to a \emph{non-additive probability} $\mu$ is defined as
\begin{equation}\label{eq:choquetGeneral}
	(\text{C}) \int_\Omega X \, d\mu \coloneqq \int_{-\infty}^0 \mu(X\geq t)-1\,dt + \int_0^{+\infty} \mu(X\geq t)\,dt,
\end{equation}
whenever it exists; see \cite[Ch. 5]{denneberg}. Choquet integration provides a natural extension to the Lebesgue integral able to deal with non-additive probabilities, as if $\mu$ in \eqref{eq:choquetGeneral} is $\sigma$-additive, then \eqref{eq:choquetGeneral} would reduce to a Lebesgue integral; see \cite{mesiar}.

We denote with $\psi(\,\cdot\,)$ a concave distortion from $[0,1]$ to $[0,1]$, i.e., a concave function such that $\psi(0)=0$ and $\psi(1)=1$, where $\psi(\mathbb{Q})(A)\coloneqq\psi(\mathbb{Q}(A))$, for every $\mathbb{Q}$-measurable set $A$; note that the distorted probability measure just defined is not, in general, additive. Let $\left(\psi_x\right)_{x\geq 0}$ be an increasing family of concave distortion functions, and assume a risk-neutral measure $\mathbb{Q}\in\bigcap_{x\geq0}\mathfrak{Q}_x$.\footnote{Given that $\mathfrak{Q}_x\subseteq\mathfrak{Q}_{x'}$ when $x'\geq x$, it is sufficient to assume that a risk-neutral measure $\mathbb{Q}$ belongs to $\mathfrak{Q}_0$.} We recall, see \cite{delbaen2009,grabisch}, that the functional $\rho_x$ such that $X\mapsto (\text{C}) \int_\Omega X \, d\psi_x({\mathbb{Q}})$ defines a coherent risk measure. This is because the (asymmetric) Choquet integral with respect to any non-additive measure guarantees the transitivity, positive homogeneity and monotonicity properties to be satisfied; see \cite[Prop. 5.1]{denneberg}. Further, the distorted probability measure $\psi_x(\mathbb{Q})$ is a submodular\footnote{A non-additive probability $\mu$ is said to be submodular (or concave) if, for every $\mu$-measurable sets $A$ and $A'$, it results that $\mu(A\cup A')+\mu(A\cap A')\leq \mu(A) + \mu(A')$.} set function, see \cite[Ex. 2.1]{denneberg}, which guarantees subadditivity; see \cite[Th. 6.3]{denneberg}. Thus, as suggested in \cite{cherny2009}, we can employ functionals of this form as tools for modeling indices of acceptability via the relationship
\begin{equation}\label{eq:oia}
\alpha(X)=\sup\left\{x\geq0:	(\text{C})\int_\Omega -X \, d\psi_x({\mathbb{Q}})\leq0\right\}.
\end{equation}
Indices of acceptability defined as in \eqref{eq:oia} are named \emph{operational indices of acceptability}; see \cite{cherny2009}.

We now assume that the market considers acceptable only the cashflows with an acceptability level of, at least, $\gamma$. The market is willing to buy $X$, which we assume to pay off at $T$, at a price $b$ if and only if $\alpha(X-\DF(T)^{-1}\cdot b)\geq\gamma$  (recall footnote \ref{fn:condition}), i.e., if and only if $
b\leq-\DF(T)\cdot(\mathrm{C})\int_\Omega -X\,d\psi_\gamma(\mathbb{Q})$. It follows that, if the market considers acceptable all the cashflows with a level of acceptability of at least $\gamma$, then the ($\gamma$-dependent) bid price of $X$, denoted as $\mathrm{bid}_\gamma(X)$, would equal
\begin{equation}\label{eq:bid}
\mathrm{bid}_\gamma(X)=-\DF(T)\cdot(\mathrm{C})\int_\Omega -X\,d\psi_\gamma(\mathbb{Q}).
\end{equation}
Denoting the ask price of $X$ given a level of acceptability $\gamma$ as $\mathrm{ask}_\gamma(X)$, by observing that $\mathrm{ask}_\gamma(X)=-\mathrm{bid}_\gamma(-X)$, from \eqref{eq:bid} it follows that
\begin{equation}\label{eq:ask}
\ask_\gamma(X)=\DF(T)\cdot(\mathrm{C})\int_\Omega X\,d\psi_\gamma(\mathbb{Q}).
\end{equation}
Thus, if the distribution function of $X$, as well as its bid or ask prices, are available, one can compute the level of $\gamma$ needed to obtain the quoted price.

\section{CDSs in a two-price economy}\label{sec:twoPriceEconomy}

Given a parametric family of distortion functions $(\psi_\gamma)_{\gamma\geq 0}$, one can set a term structure for the liquidity parameter $\gamma$ by assigning a value $\gamma_i$ to each maturity $e_{N_i}$. These values can be then interpolated, once the model has been calibrated, if one wants to calculate bid and ask prices for non-quoted maturities. We still assume a Poisson process as in Section \ref{dynamicsOfTheSurvivalProbabilities} driving the risk-neutral dynamics of $\tau$. We denote with $\tilde{X}_i^{\text{CDS}}$ the sum of the cashflows of the $i$\textsuperscript{th} CDS where all the cashflows are deferred to maturity, i.e., $\text{cashflow}(t)\mapsto \text{cashflow}(t)\cdot\frac{\DF(t)}{\DF(e_{N_i})}$. The bid and ask prices of $\tilde{X}_i^{\text{CDS}}$ are denoted as $\text{bid}^{\text{CDS}}_i$ and $\text{ask}^{\text{CDS}}_i$, respectively. 

We start with the first CDS. We need then to solve for $\lambda_1>0$ and $\gamma_1>0$ such that
\begin{equation}
\begin{cases}
\text{bid}^{\text{CDS}}_1(\lambda_1, \gamma_1)+\mathrm{Acc} = \DF(t_s)\cdot\UF_1^\text{bid}\\
\text{ask}^{\text{CDS}}_1(\lambda_1, \gamma_1) +\mathrm{Acc} = \DF(t_s)\cdot \UF_1^\text{ask}
\end{cases}\!\!\!\!\!\!,
\end{equation}
with
\begin{equation}\label{eq:constraint}
 \DF(t_s)\cdot\UF_1^{\text{bid}} < \PV^{\text{CDS}}_1(\lambda_1)+\mathrm{Acc}<  \DF(t_s)\cdot\UF_1^{\text{ask}},
\end{equation}
where the constraint (\ref{eq:constraint}) guarantees that the risk-neutral price of the CDS lies between its corresponding bid and ask prices.

By proceeding in a similar way as done in Section \ref{dynamicsOfTheSurvivalProbabilities}, at every step we need to solve a system of the form
\begin{equation}
	\begin{cases}
	\text{bid}_{i}(\lambda_i,\gamma_i) + \mathrm{Acc} = \DF(t_s)\cdot\UF_i^\text{bid}\\
	\text{ask}_{i}(\lambda_i,\gamma_i) + \mathrm{Acc} = \DF(t_s)\cdot\UF_i^\text{ask}
	\end{cases}\!\!\!\!\!\!,
\end{equation}
	with
	\begin{equation}
		{\DF(t_s)\cdot\UF_i^{\text{bid}} < \PV^{\text{CDS}}_{i}(\lambda_i) +\mathrm{Acc} < \DF(t_s)\cdot\UF_i^{\text{ask}}}.
	\end{equation}
Above, $\lambda_i$ represents the implied hazard rate for the $i$\textsuperscript{th} maturity, while $\gamma_i$ the corresponding \emph{implied liquidity} in the sense of \cite{impliedLiquidity}.	
	
The problem of determining whether a (potentially unique) solution for this constrained non-linear system will be addressed in this section. We start by simplifying the notation in the constrained system above by rewriting it as  
\begin{equation}
	\begin{cases}\label{eq:system}
	\bidarg = b\\
	\askarg = a
	\end{cases}\!\!\!\!\!\!,
	\end{equation}
	with
	\begin{equation}\label{eq:constraint2}
		b < \PV^{\text{CDS}}(\lambda) < a.
	\end{equation}
In both (\ref{eq:system}) and (\ref{eq:constraint2}) we have set $b\coloneqq \DF(t_s)\cdot\UF_i^{\text{bid}}-\mathrm{Acc}$ and $a\coloneqq \DF(t_s)\cdot\UF_i^{\text{ask}} -\mathrm{Acc}$.

\medskip

We provide now three lemmas that, under some mild assumptions related to the liquidity of the market, will be used in Theorem \ref{prop:existence} to prove the existence and the uniqueness of a solution for the constrained non-linear system (\ref{eq:system}). We start by assuming that the quoted bid and ask prices of the chosen CDS are within the interval of possible risk-neutral prices that can be obtained by changing the parameter $\lambda$. In practice, this technical condition translates into the possibility of being able to calibrate the risk-neutral parameter $\lambda$ to match bid and ask market quotes, respectively, from which Lemma \ref{lemma:lambda interval} follows.

\begin{Assumption} The inequalities $\inf_{\lambda> 0} \PV^{\mathrm{CDS}}(\lambda) < b$ and $\sup_{\lambda> 0} \PV^{\mathrm{CDS}}(\lambda) > a$ hold.
\end{Assumption}

\begin{lemma}\label{lemma:lambda interval}
Under Assumption 1, there exists an interval $[\lambda_b, \lambda_a]$ such that there is equivalence between $b\leq \PV^{\mathrm{CDS}}(\lambda) \leq a$ and $\lambda \in [\lambda_b, \lambda_a]$.
\end{lemma}
\begin{proof}
$\PV^{\mathrm{CDS}}(\lambda)$ is an increasing and continuous function of $\lambda$. From Assumption 1 the result follows.
\end{proof}

We now introduce a second assumption that guarantees that, for $\lambda$ in a given range, theoretical bid-ask spreads can exceed the observed one.\footnote{Observe that, when $\gamma=0$, then bid and ask prices reduce to the ones calculated with respect to $\mathbb{Q}$ and that, for a given $\lambda$, the function $\askarg-\bidarg$ is strictly increasing in $\gamma$. When $\gamma\rightarrow +\infty$, then $\psi_\gamma(\mathbb{Q})$ approximate the distribution that assigns zero to the null sets and one to any other set. $\askarg-\bidarg$ can be rewritten, see Section \ref{sec:example}, as $\askargprot - \bidargprem - \bidargprot +\askargprem$, where the superscripts identify the two legs of the contract. Ignoring discount factors for simplicity, $\askargprot$ has magnitude of the order of $\lgd\cdot\psi_\gamma(\mathbb{Q})(t_p\leq \tau\leq e_N)=\lgd$; see (\ref{eq:pvProt}). Further, as $(\text{C})\int_\Omega- X \, d\mu=(\text{C})\int_\Omega X \, d\overbar{\mu}$ with $\overbar{\mu}$ denoting the dual measure of $\mu$, see \cite[Prop. 5.1]{denneberg}, from $\bid(X)=-\ask(-X)$ it  follows that $\bidargprot=-\lgd\cdot\overline{\psi_\gamma(\mathbb{Q})}(t_p\leq \tau\leq e_N)=0$. Therefore, for extreme values of $\gamma$ the theoretical bid-ask spread $\askarg-\bidarg$ reaches high values due to its positive components $\askargprot$ and $\askargprem$, and to $\bidargprem$ being below its counterpart calculated when $\gamma=0$. Thus, for practical purposes, this assumption is in general satisfied.} Intuitively, Assumption 2 is a technical condition stating that, for every fixed $\lambda$ in $[\lambda_a,\lambda_b]$, it is always possible to construct a bid and an ask price that reflect the bid-ask spread observed in the market. Lemmas \ref{lemma:uniqueness gamma} and \ref{lemma:continuity gamma} follow.
\begin{Assumption} For every $\lambda \in [\lambda_b, \lambda_a]$ there exists $\gamma> 0$ such that $\askarg - \bidarg = a-b$.
\end{Assumption}

	\begin{lemma}\label{lemma:uniqueness gamma}
	Under Assumptions 1 and 2, for every $\lambda \in [\lambda_b, \lambda_a]$ there exists a unique $\gamma>0$ such that $\askarg - \bidarg = a-b$.
	\end{lemma}
	\begin{proof}
	Fix $\lambda \in [\lambda_b, \lambda_a]$. By Assumption 2 there exists (at least) one $\gamma>0$ such that $\askarg - \bidarg = a-b$. Assume there exists $\gamma_*$ and $\gamma^*$ such that $\ask(\lambda,\gamma_*)-\bid(\lambda,\gamma_*)=\ask(\lambda,\gamma^*)-\bid(\lambda,\gamma^*)=a-b$, with $\gamma_*<\gamma^*$. The ask price is an increasing function of $\gamma$, while the opposite holds for the bid. Therefore, $\ask(\lambda,\gamma_*)<\ask(\lambda,\gamma^*)$ and $\bid(\lambda,\gamma_*)>\bid(\lambda,\gamma^*)$. Then $a-b=\ask(\lambda,\gamma_*)-\bid(\lambda,\gamma_*)<\ask(\lambda,\gamma^*)-\bid(\lambda,\gamma^*)=a-b$, contradiction.
	\end{proof}
	
\begin{lemma}\label{lemma:continuity gamma}
Under Assumptions 1 and 2, for every $\lambda \in [\lambda_b, \lambda_a]$ the function such that $\lambda\mapsto \gamma(\lambda)$, where $\mathrm{ask}(\lambda, \gamma(\lambda))-\mathrm{bid}(\lambda, \gamma(\lambda))=a-b$, is continuous.
\end{lemma}
\begin{proof}
Fix $\bar{\lambda}$ in $[\lambda_b, \lambda_a]$ and let $(\lambda_n)_n$ be a sequence in $[\lambda_b, \lambda_a]$ that converges to $\bar{\lambda}$. We define $\phi(\lambda,\gamma)\coloneqq\askarg - \bidarg$. We proceed in steps.

Claim (i): The sequence $(\gamma(\lambda_n))_n$ is bounded. Say this is not the case. Then, there exists a subsequence $(\gamma(\lambda_{n_k}))_k$ that diverges to $+\infty$. $(\lambda_{n_k})_k$ converges to $\bar{\lambda}$, as subsequence of a convergent sequence, and $\phi$ is continuous in both arguments. Therefore, $\lim_k \phi(\bar{\lambda}_{n_k}, \gamma(\lambda_{n_k}))=\phi(\bar{\lambda}, +\infty)=a-b$, as $\phi(\bar{\lambda}_{n_k}, \gamma(\lambda_{n_k}))$  always equals $a-b$, by construction. By Assumption 2 there exists $\bar{\gamma}> 0$ such that $\phi(\bar{\lambda}, \bar{\gamma})=a-b$. Therefore, as $\phi$ is increasing in its second argument, it follows that $a-b = \phi(\bar{\lambda}, \bar{\gamma})<\phi(\bar{\lambda}, +\infty)=a-b$, contradiction.

Claim (ii): The sequence $(\gamma(\lambda_n))_n$ has limit. As this sequence is bounded, it admits a convergent subsequence. Say there are two subsequences, namely $(\gamma(\lambda_{n_k}))_k$ and $(\gamma(\lambda_{n_h}))_h$, that converge to $\gamma_*$ and $\gamma^*$, respectively, where $\gamma_*<\gamma^*$. Then $(\lambda_{n_k})_k$ and $(\lambda_{n_h})_h$ both converge to $\bar{\lambda}$, as subsequencies of the same convergent sequence. So we obtain that $a-b = \lim_k\phi(\lambda_{n_k},\gamma(\lambda_{n_k}))=\phi(\bar{\lambda},\gamma_*) <\phi(\bar{\lambda},\gamma^*) = \lim_h \phi(\lambda_{n_h},\gamma(\lambda_{n_h}))=a-b$, contradiction (the first and the last equalities follow from the definitions of $(\lambda_{n_k})_k$ and $(\lambda_{n_h})_h$, respectively, the second and the penultimate equalities from the continuity of $\phi$, while the inequality from $\phi$ being increasing in its second argument). Then, every convergent subsequence of $(\gamma(\lambda_n))_n$ has the same limit. As $(\gamma(\lambda_n))_n$ is bounded, then it admits limit.\footnote{Here, we have used the following elementary result: if a bounded real sequence has the property that all its convergent subsequences converge to the same real limit, then the sequence itself also converges to it; see \cite[Ex. 2.5.5]{Abbott}.}

Claim (iii): The limit of $(\gamma(\lambda_n))_n$ is $\gamma(\bar{\lambda})$. Denote $\lim_n \gamma(\lambda_n)$ as $\bar{\gamma}$. Observe that $\phi$ is continuous in both arguments. The sequence $(\phi(\lambda_n, \gamma(\lambda_n)))_n$ is constant by construction, i.e., it always equals $a-b$. Therefore, it converges to $a-b$. Its limit is $\phi(\bar{\lambda},\bar{\gamma})$, as $\phi$ is continuous. Due to Lemma \ref{lemma:uniqueness gamma}, there exists a unique $\gamma(\bar{\lambda})$ such that $\phi(\bar{\lambda}, \gamma(\bar{\lambda}))=a-b$. So, $\bar{\gamma}=\gamma(\bar{\lambda})$.
\end{proof}

We now can, under Assumptions 1 and 2, use Lemmas \ref{lemma:lambda interval}, \ref{lemma:uniqueness gamma} and \ref{lemma:continuity gamma} to prove that the calibration problem \eqref{eq:system} has a unique solution. Therefore, Theorem \ref{prop:existence} guarantees that, under the hypotheses considered, risk-neutral default probabilities can be implied in a unique manner from bid and ask CDS quotes without relying on their mid counterparts.

\begin{theorem}\label{prop:existence}
		Under Assumptions 1 and 2, there exists a solution of the constrained non-linear system \eqref{eq:system}, and it is unique.
	\end{theorem}
	\begin{proof}
	Consider the interval $[\lambda_b, \lambda_a]$ as per Lemma \ref{lemma:lambda interval}. There exists a unique $\gamma_b$ such that $\ask(\lambda_b,\gamma_b)-\bid(\lambda_b,\gamma_b)=a-b$. Observe that $\bid(\lambda_b,\gamma_b)<\PV^{\text{CDS}}(\lambda_b)=b$, so $b<\ask(\lambda_b,\gamma_b)<a$. 
	
Similarly, consider $\lambda_a$. There exists a unique $\gamma_a$ such that $\ask(\lambda_a,\gamma_a)-\bid(\lambda_a,\gamma_a)=a-b$. Because $a=\PV^{\text{CDS}}(\lambda_a)<\ask(\lambda_a,\gamma_a)$, it follows that $b<\bid(\lambda_a,\gamma_a)<a$.

The functions $\ask(\lambda, \gamma)$, $\bid(\lambda, \gamma)$, and -- see Lemma  \ref{lemma:continuity gamma} -- $\gamma(\lambda)$, are continuous in $\lambda$. Thus, there exists $\bar{\lambda}\in (\lambda_b, \lambda_a)$ and corresponding $\bar{\gamma}$ such that $\ask(\bar{\lambda},\bar{\gamma})=a$ and $\bid(\bar{\lambda},\bar{\gamma})=b$. By virtue of Lemma \ref{lemma:uniqueness gamma} the pair $(\bar{\lambda},\bar{\gamma})$ satisfying (\ref{eq:system}) is unique.
	\end{proof}
	
Note that to obtain the existence and uniqueness result of Theorem \ref{prop:existence} we have relied on the fact that, for each given  maturity, the model describing the risk-neutral default distribution has a single free parameter, i.e., the hazard rate corresponding to the maturity considered. Therefore, considering the distortion parameter related to that maturity as additional degree of freedom allows the calibration problem to be defined, up to the constraint, by two equations and two unknowns. If more complex models with additional parameters were to be used, then the problem should have been approached in a least-square sense, and the best  possible outcome would have been that of finding an unique minimum. This very favorable situation, however, would not necessarily guarantee observed market quotes to be matched by the model, and therefore would as well not guarantee the implied risk-neutral distribution to be the ``true'' one.

\subsection{A calibration example}\label{sec:example}
The simplest possible manner to specify model (\ref{lambda}) consists in defining it as a piecewise-constant function, as done in the ISDA CDS  standard model commonly used in practice, which is based on the approach of \cite{okaneTurnbull}. Note that there is no information available on the hazard rate level between CDS maturities. Therefore, these specifications provide the smallest possible set of assumptions concerning the default intensity process and is a common choice amongst financial practitioners.

To compute bid and ask prices, one would need to approximate Choquet integrals numerically. To do so, a simple approximation of $(\text{C}) \int_\Omega X \, d\mu$ can be performed, see \cite[Ch. 11.5]{generalizedMeasureTheory}, as follows. Given a partition of $\Omega$ as $\bigcup_{i=1}^M A_i$ choose, for every $i$, $x_i\in X(A_i)$. Let $\sigma$ denote a permutation of $\left\{1,\ldots,M\right\}$ such that $x_{\sigma(1)}\leq\ldots\leq x_{\sigma(M)}$. Then, $(\text{C}) \int_\Omega X \, d\mu$ can be then approximated as
\begin{equation}
\sum_{i=1}^M(x_{\sigma(i)}-x_{\sigma(i-1)})\cdot \mu\left( \bigcup_{k=i}^M A_k\right),
\end{equation}
where $x_0\coloneqq 0$. In the case of a CDS, one can then set a grid (for instance, daily for simplicity),   namely $A_1\coloneqq\left\{\tau\in[0,d_1]\right\},\ldots,A_M\coloneqq\left\{\tau\in[d_{M-1},d_{M}]\right\}$, where $M$ denotes the total number of points (i.e., dates) in the grid, and set $x_i:=\tilde{X}^\text{CDS}|_{\tau=d_i}$ (recall that, using the notation introduced in Section \ref{sec:twoPriceEconomy}, the superscript tilde indicates that cashflows are deferred at the maturity of the CDS contract considered).

We recall, see \cite{eberlein2012}, that the bid and ask prices of a contingent claim $X$ can be calculated as $\bid(X)=\bid(X^+)-\ask(X^-)$ and $\ask(X)=\ask(X^+)-\bid(X^-)$, respectively, where the $X^+$ ($X^-$) denotes the positive (negative) part of $X$. We denote with $\tilde{X}^{\text{prot}}$ ($\tilde{X}^{\text{prem}}$) the protection (premium) leg of  $\tilde{X}^{\text{CDS}}$. From (\ref{eq:pv}), and by noting that with our conventions $X^+$ coincides with $\tilde{X}^{\text{prot}}$ and $X^-$ with $\tilde{X}^{\text{prem}}$, it follows that $\bid(\tilde{X}^{\text{CDS}})=-\ask(-\tilde{X}^{\text{prot}})-\ask(\tilde{X}^{\text{prem}})$ and that $\ask(\tilde{X}^{\text{CDS}})=\ask(\tilde{X}^{\text{prot}})+\ask(-\tilde{X}^{\text{prem}})$. Therefore, in principle it is sufficient to separately calculate the ask prices of the (signed) CDS legs only.

\medskip

As an example, we consider the specifications of the standard ISDA model, i.e., we assume a piecewise-constant hazard rate function. We take into account a set of market quotes for a BBB European financial institution with maturities 6 months and 1, 2, 3, 4, 5, 7 and 10 years, respectively, as of 13 February 2020. The recovery rate equals 40\%, and the coupon 1\%. Discounting performed with OIS EUR curve.

Figure \ref{fig:premia} represents the bid and ask quoted upfront premia, expressed per unit of notional, while Figure \ref{fig:calibrationError} the aggregated calibration errors, i.e., each value represents the sum of the bid and ask calibration errors, respectively.

\begin{figure}[H]
\begin{center}
   \begin{subfigure}{.5\linewidth}
     \centering
     \includegraphics[scale=0.375]{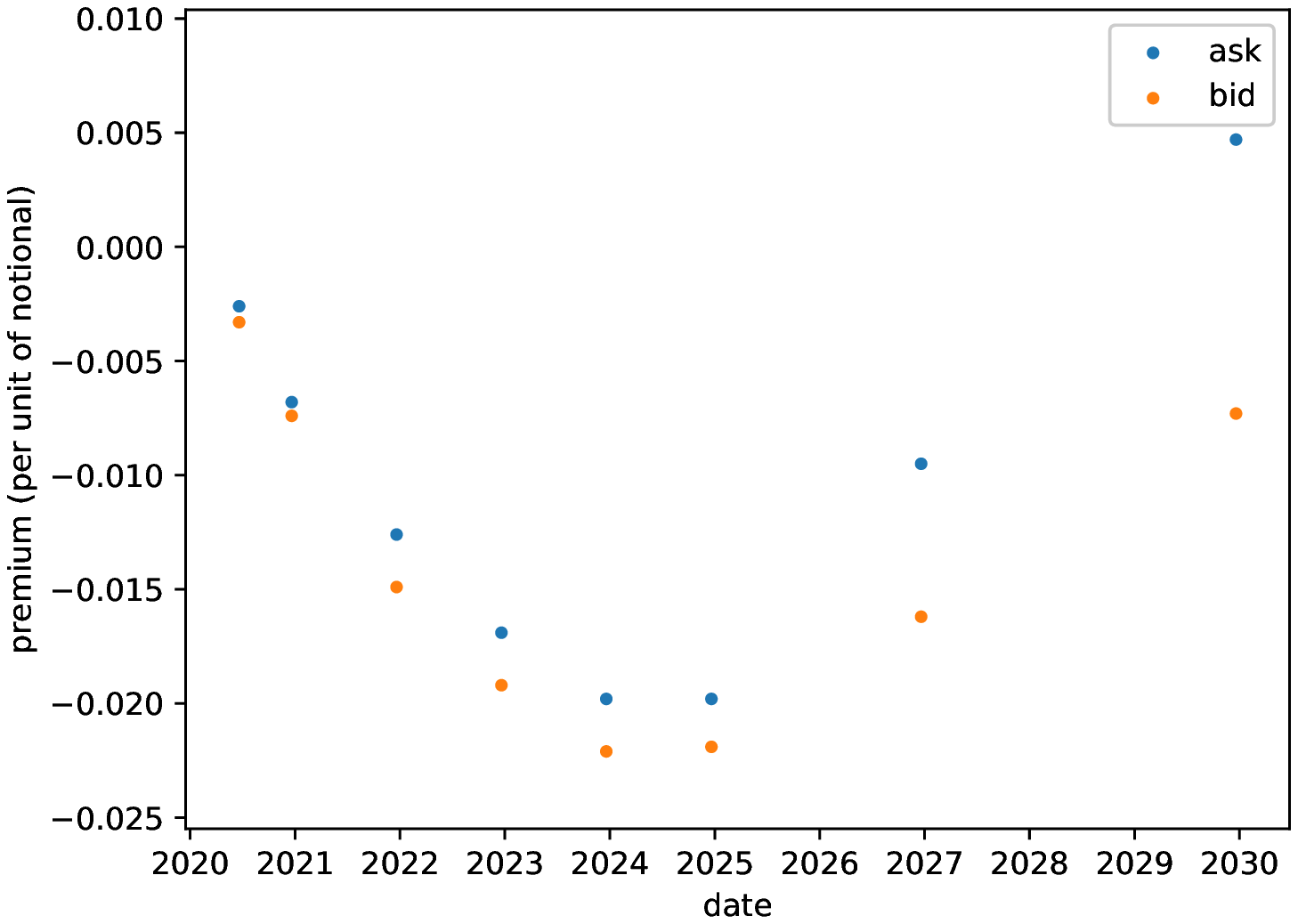}
     \caption{}\label{fig:premia}
   \end{subfigure}\hspace*{-0.5cm}
   \begin{subfigure}{.5\linewidth}
     \centering
     \includegraphics[scale=0.375]{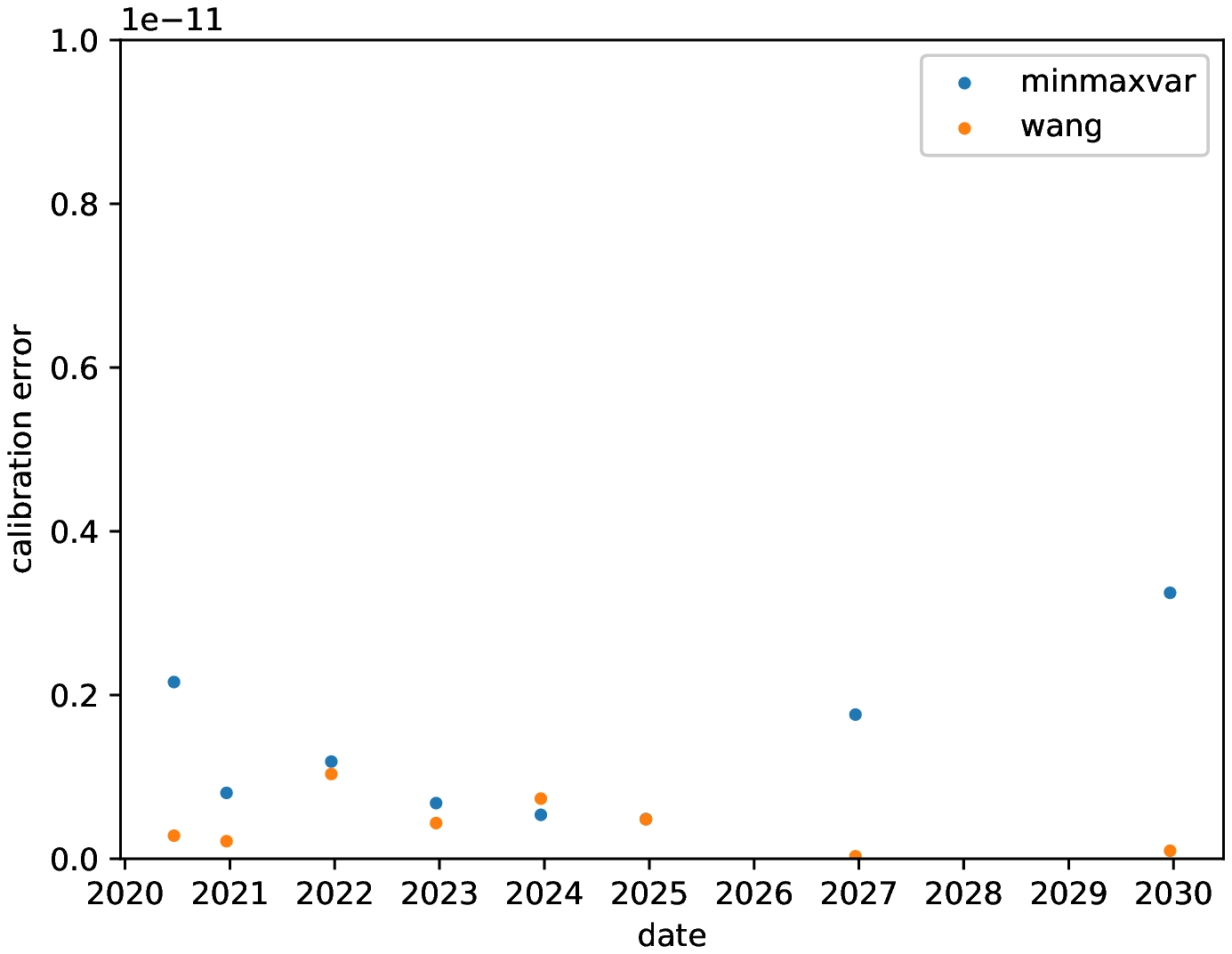}
     \caption{}\label{fig:calibrationError}
   \end{subfigure}
   \caption{Bid and ask CDS upfront premia used in the calibration example, panel (a), and calibration errors, panel (b). In particular, the calibration errors of panel (b) represent aggregated figures, i.e., each point corresponds to the sum of the calibration error for the bid quote and of that of the related ask quote.}
   \end{center}
\end{figure}

In this example we consider two common choices to define the family of distortion function, i.e., the \emph{minmaxvar} distortion \cite{cherny2009}, defined via
\begin{equation}\label{eq:minmaxvar}
\psi_\gamma(x)\coloneqq1-\left(1-x^{\frac{1}{1+\gamma}}\right)^{1+\gamma},
\end{equation}
and the \emph{Wang} distortion \cite{wang}, defined by setting
\begin{equation}\label{eq:wang}
\psi_\gamma(x)\coloneqq\Phi\left(\Phi^{-1}(x) + \gamma\right),
\end{equation}
with $\Phi(\,\cdot\,)$ denoting the cumulative distribution function of a standard normal random variable; in both \eqref{eq:minmaxvar} and \eqref{eq:wang} it is assumed that $x\in[0,1]$ (in the case of the latter, right and left limit should be considered for 0 and 1, respectively) and that $\gamma\geq 0$. Other examples of families of distortion functions are outlined, for instance, in \cite[Ch. 4.7]{bookConicFinance} and in \cite[Ch. 4.6]{follmer}. For each of the two choices we have made in terms of the distortion function, Figure \ref{fig:lambda} shows the piecewise-constant hazard rate function, while Figure \ref{fig:gamma} the linearly interpolated distortion parameter. Note that for each of the two choices of the distortion function we have made, the minimum of the $\gamma$ parameter in Figure \ref{fig:gamma} lies in proximity of the 5Y CDS, where the latter is usually the most liquid maturity. We also note how the pattern of the implied distortion parameter in Figure \ref{fig:gamma} follows that of the (relative) bid-ask CDS premium spread available in Table \ref{tab:premia} (last column therein). 
\begin{figure}[H]
\begin{center}
   \begin{subfigure}{.5\linewidth}
     \centering
     \includegraphics[scale=0.375]{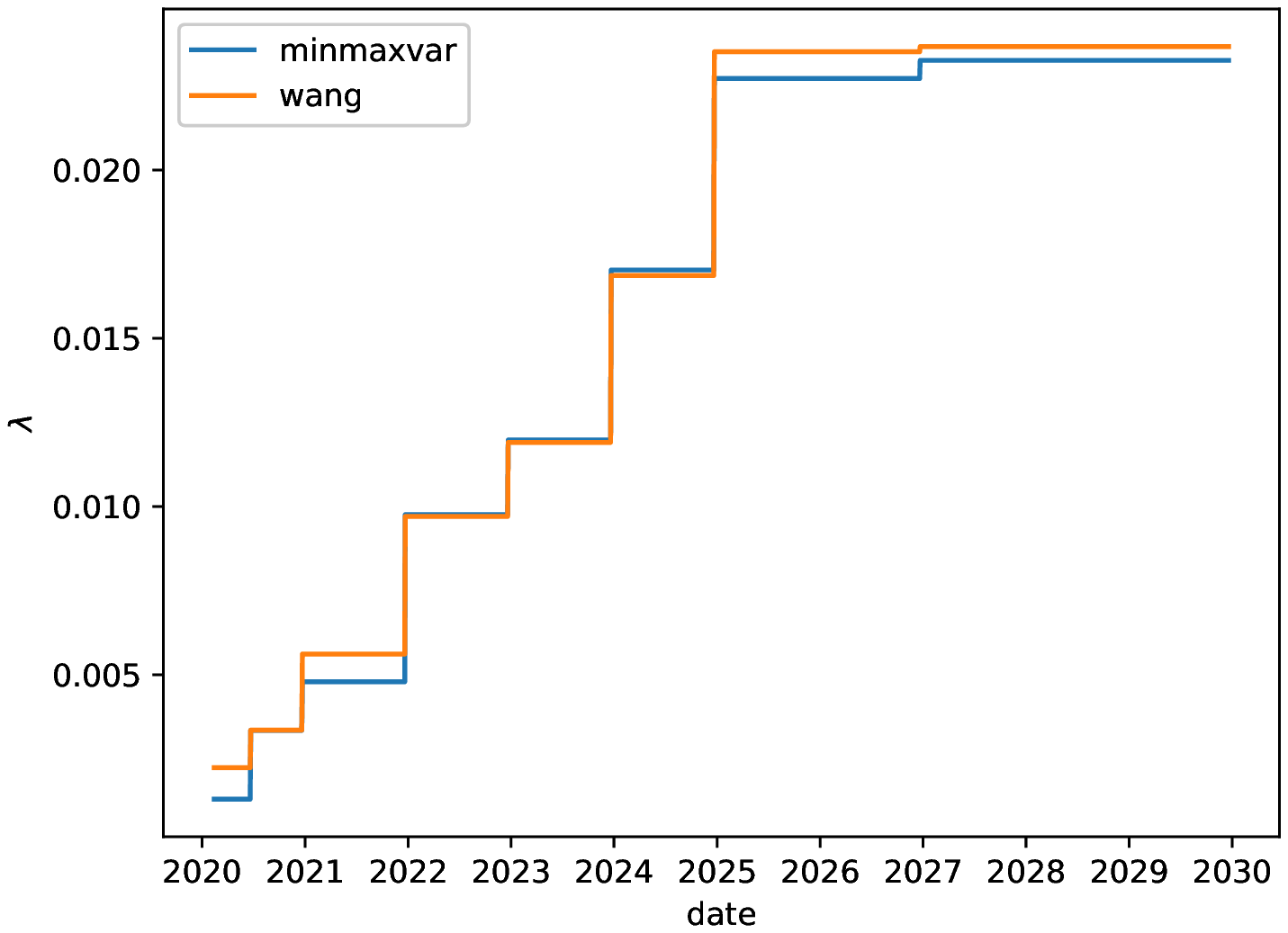}
     \caption{}\label{fig:lambda}
   \end{subfigure}\hspace*{-0.5cm}
   \begin{subfigure}{.5\linewidth}
    \centering
     \includegraphics[scale=0.375]{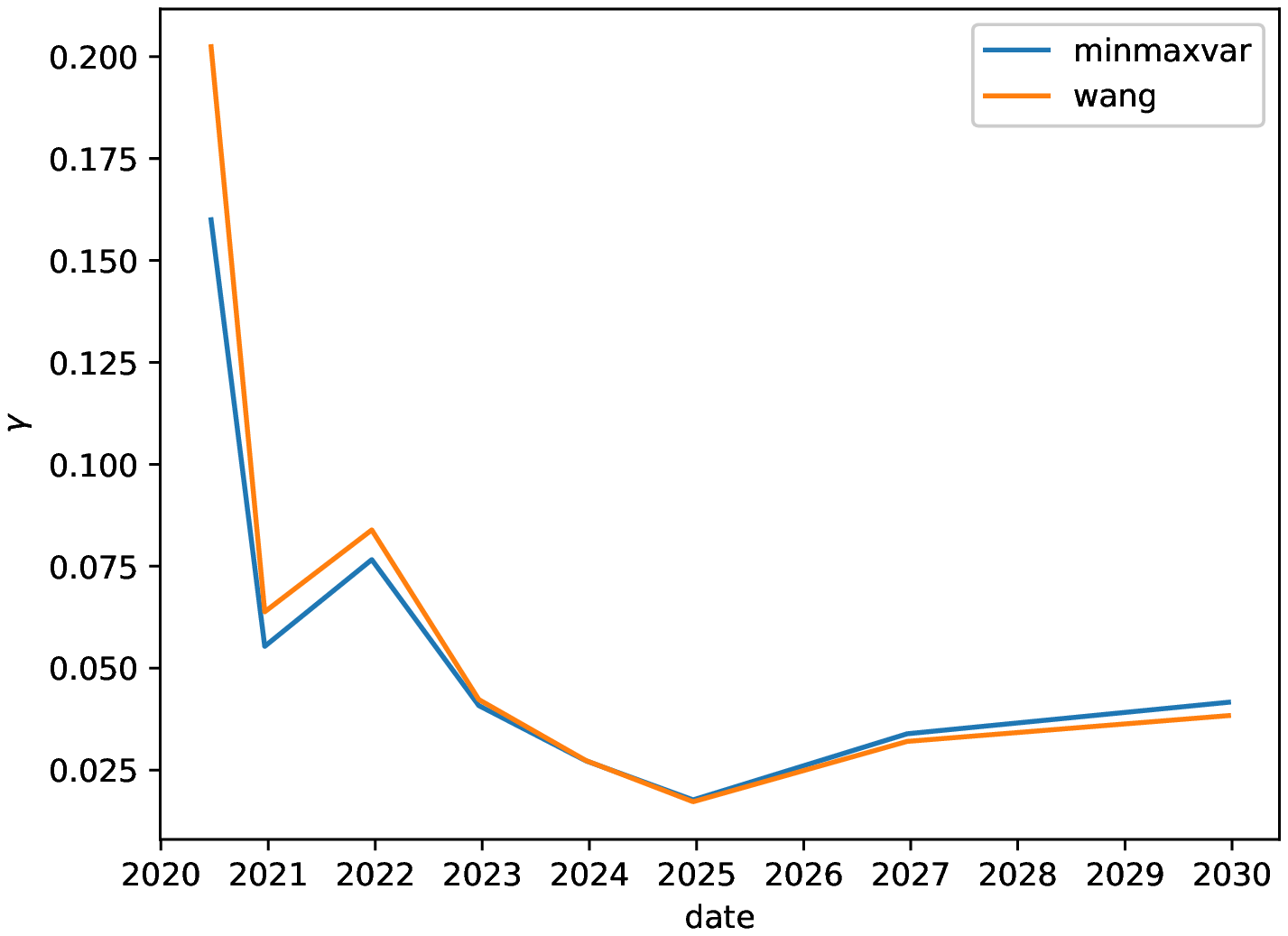}
     \caption{}\label{fig:gamma}
   \end{subfigure}
   \end{center}
   \caption{Implied parameters provided as a result of the calibration procedure to the bid and ask CDS quotes of Figure \ref{fig:premia}: hazard rates (i.e., $\lambda$) are depicted in panel (a), while distortion parameters (i.e., $\gamma$) in panel (b).}
\end{figure}
Note that, as illustrated by Figures \ref{fig:lambda} and \ref{fig:gamma} results obtained using the minmaxvar and Wang transforms are very similar, indicating little model risk. For completeness, the CDS input quotes of Figure \ref{fig:premia}, the implied hazard rates of Figure \ref{fig:lambda} and the implied distortion parameters of Figure \ref{fig:gamma} have been reported in \ref{sec:tables}; see Tables \ref{tab:premia}, \ref{tab:lambdas} and \ref{tab:gammas}, respectively.

\section{Conclusion}\label{sec:conclusion}
In this article we have considered the problem of calibrating a CDS model to the available bid and ask quotes within the conic finance paradigm of \cite{cherny}. In particular, in the context of reduced-form models, we have considered the default time as modeled by a Poisson process. The bid-ask calibration problem requires to iteratively solve a constrained non-linear system in two equations and two unknowns. We have showed that, under reasonable assumption for practical purposes, the calibration problem admits a unique solution. We have as well illustrated, with a practical example based on real market data, how the calibration to bid and ask CDS quotes works under the specifications of the standard CDS ISDA model and by considering two different choices for the the distortion function. In both the cases considered, as expected from the theory, the model could be calibrated to exactly match the observed market quotes. Despite our work outlined in Section \ref{sec:twoPriceEconomy} is specific to CDSs, the fact that financial instruments trade neither at the risk-neutral nor at the mid price apply to all contingent claims. Therefore, being able to fit valuation models solely to bid and ask quotes in such a way that risk-neutral parameters are implied as a result of the calibration routines is a desirable model feature that allows to drop the common assumption of equating risk-neutral and mid prices without additional ones being added. Hence, investigating how to calibrate models to bid and ask quotes without relying on mid quote approximations is a field on which further research is encouraged.
\newpage
\clearpage
\appendix

\section{A remark on  the monotonicity of CDS prices}\label{sec:remark}
We consider here the $i$\textsuperscript{th} CDS outlined in Section \ref{dynamicsOfTheSurvivalProbabilities}, i.e., the one maturing at $e_{N_i}$, and we denote with $N(i)$ the number of coupon periods related to it.

From \eqref{eq:pvPrem}, the present value of its premium leg can be rewritten as
\begin{equation}\label{eq:pvPrem2}
C\sum_{j=1}^{N(i)} \left(\DF(t_j)\cdot \Delta_j \cdot  \PS(e_j)+\mathbb{E}^{\mathbb{Q}}\left(\DF(\tau) \cdot \Delta(s_j, \tau) \cdot  \mathds{1}_{\left\{s_j\leq \tau\leq e_j\right\}}\right)\right).
\end{equation}
We define $j(i)\coloneqq \min\left\{j: e_j> e_{N_{i-1}},\, 1\leq j\leq N(i)\right\}$, with the convention that $j(i)=1$ if $i=1$. If $\lambda_i$ increases, from \eqref{PSlambda} and \eqref{lambda} it follows that $\PS(e_j)$ strictly decreases for each $j\geq j(i)$, leaving the others, if any, unchanged.
We can also rewrite the present value of the protection leg, see \eqref{eq:pvProt}, minus the accrual payments in \eqref{eq:pvPrem2}, as
\begin{align}
&\!\!\!\!\!\!\!\!\!\!\!\!\!\sum_{j=1}^{N(i)} \mathbb{E}^{\mathbb{Q}}\left(\DF(\tau)\cdot\left(\lgd- C\cdot\Delta(s_j, \tau)\right) \cdot  \mathds{1}_{\left\{\max(t_p,s_j)\leq \tau\leq e_j\right\}}\right)\nonumber \\
&\ \ \ \ \ \ \ \ \ \ \ \ \ \ \ \ \ \ \ \ \ \ \ \ \ - \mathbb{E}^{\mathbb{Q}}\left(\DF(\tau)\cdot C\cdot\Delta(s_1, \tau) \cdot  \mathds{1}_{\left\{0\leq \tau\leq t_p\right\}}\right),\label{eq:pvProt2}
\end{align}
due to $s_j=e_{j-1}$ whenever $j>1$. 

If $i>1$, when $\lambda_i$ increases then $\mathbb{Q}(\max(t_p,s_j)\leq \tau\leq e_j)$ strictly increases for each $j\geq j(i)$, leaving the other probabilities with $j<j(i)$, as well as $\mathbb{Q}(0\leq \tau\leq t_p)$, unchanged. Thus, if the condition
\begin{equation}\label{eq:condition}
\lgd>C\cdot\max_{j(i)\leq j\leq N(i)}\Delta(s_j, e_j)
\end{equation}
holds, then \eqref{eq:pvProt2} strictly increases if $\lambda_i$ increases.

In practice, condition \eqref{eq:condition} is verified for usual values of $\lgd$ and $C$: for instance, if the often-standard value for $\lgd$ of $60\%$ is chosen and $C=5\%$, then the right-hand side of \eqref{eq:condition} would be equal, up to day-count rounding, to $5\%\cdot0.25=1.25\%$, due to the quarterly payments of each CDS contract. A graphical illustration is provided in Figure \ref{fig:PV6}.
\begin{figure}[H]
\centering
     \includegraphics[scale=0.375]{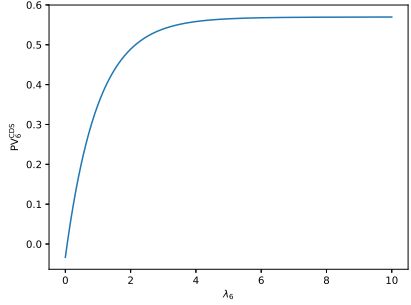}
      \caption{Present value of the 6\textsuperscript{th} CDS used for the calibration example provided in Section \ref{sec:example} (5Y maturity) with risk-neutral default probabilities implied via the minmaxvar distortion, as a function of $\lambda_6$. Notional assumed unitary.}
      \label{fig:PV6}
\end{figure}

Note that, when $i=1$, the second summand in \eqref{eq:pvProt2} is, in general, negligible. This is because the event $\left\{0\leq \tau\leq t_p\right\}$ means observing a default between the valuation date and the protection start date of the CDS, where the latter is usually one day after the former. With a good approximation $\DF(\tau)$ equals 1, as when $0\leq \tau\leq t_p$ the year fraction between the valuation date and the default time is almost zero. Further, we can bound $\Delta(s_1,\tau)$ from above, up to day-count rounding, with 0.25. Thus, an approximate upper bound for $\mathbb{E}^{\mathbb{Q}}\left(\DF(\tau)\cdot C\cdot\Delta(s_1, \tau) \cdot  \mathds{1}_{\left\{0\leq \tau\leq t_p\right\}}\right)$ is given by $ C\cdot 0.25\cdot\mathbb{Q}\left( 0\leq\tau\leq t_p\right)$. To give an idea about the magnitude of this term, if we consider as a simple case a piecewise-constant hazard rate functional form for  \eqref{lambda}, we than have that $\mathbb{Q}\left(0\leq \tau \leq t_p\right)= 1-e^{-\lambda_1\cdot\Delta(0,t_p)}$. If $\lambda_1$ increases by an amount $\delta$, by using a first order Taylor expansion we obtain that $\mathbb{Q}\left(0\leq \tau \leq t_p\right)$ increases by approximately $\delta\cdot\Delta(0,t_p)$. Thus, if $\lambda_1$ increases by $\delta$ then the change in $\mathbb{E}^{\mathbb{Q}}\left(\DF(\tau)\cdot C\cdot\Delta(s_1, \tau) \cdot  \mathds{1}_{\left\{0\leq \tau\leq t_p\right\}}\right)$ is approximately bounded from above by $C\cdot0.25\cdot \delta\cdot\Delta(0,t_p)$. Again, assume $C$ equals $5\%$ and that $t_p$ occurs one day after the valuation date. Using the Act/360 day-count convention we obtain that this amount equals $1.25\%\cdot\delta\cdot\frac{1}{360}$, which is negligible when $\delta$ not too large; see Figure \ref{fig:PV1} for an example.

When $i=1$, if $\lambda_1=0$ then the present value of the protection leg would be zero, making the value of the contract negative. When $\lambda_1$ increases, the present value of the contract increases as well, and for $\lambda_1$ large enough it would reach a positive value. However, when $\lambda_1$ diverges to $+\infty$, then a default would occur while the contract is being signed, which would make the value of the contract drop. Therefore, the monotonicity would be guaranteed, when $i=1$ and when usual coupon and $\lgd$ amounts are considered, on an interval $[0,\tilde{\lambda}_1]$, which is usually wide enough for practical applications. This is illustrated in Figure \ref{fig:PV1_extreme}.
\begin{figure}[H]
\begin{center}
   \begin{subfigure}{.5\linewidth}
     \centering
     \includegraphics[scale=0.375]{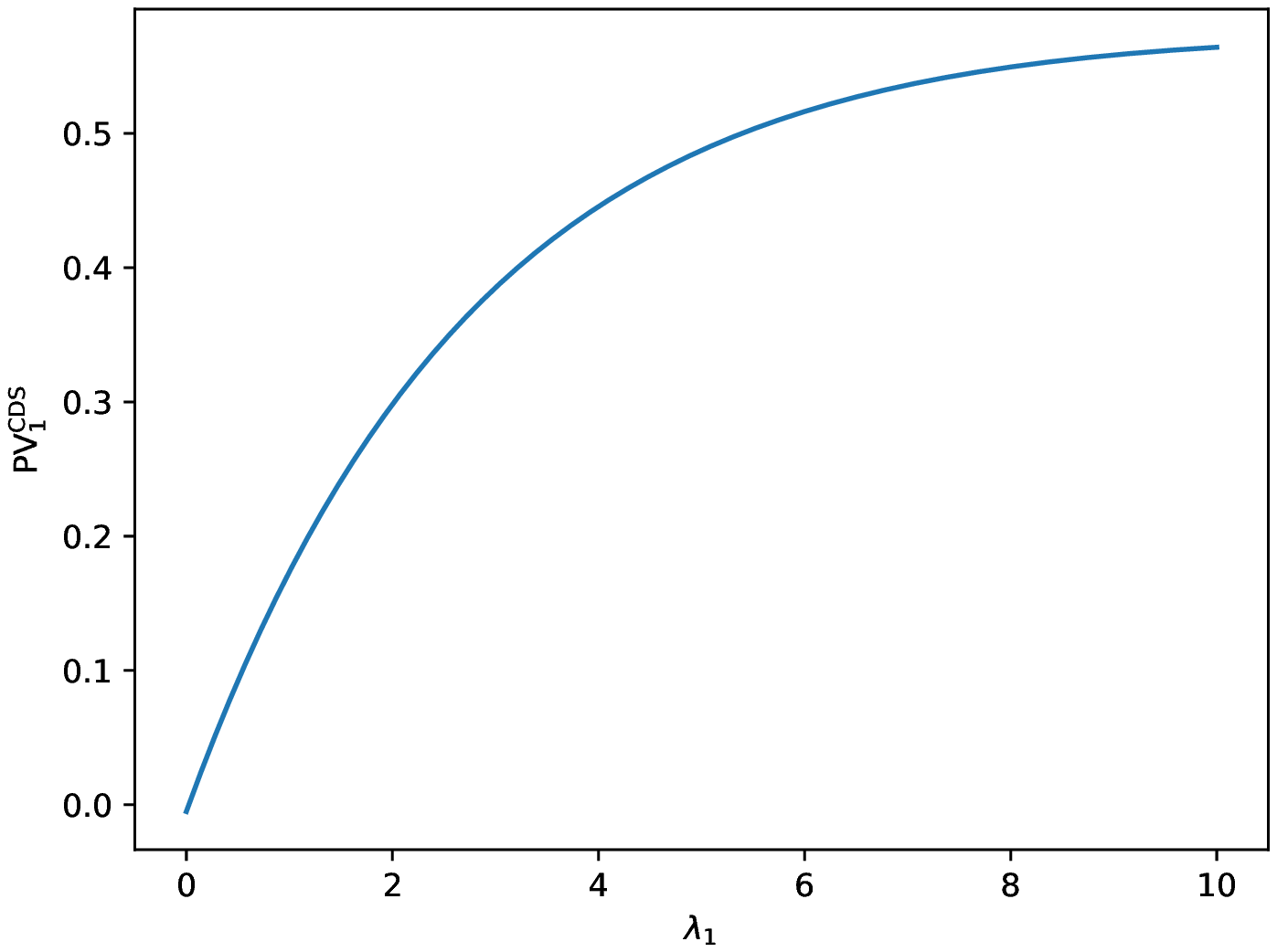}
     \caption{}\label{fig:PV1}
   \end{subfigure}\hspace*{-0.5cm}
   \begin{subfigure}{.5\linewidth}
    \centering
     \includegraphics[scale=0.375]{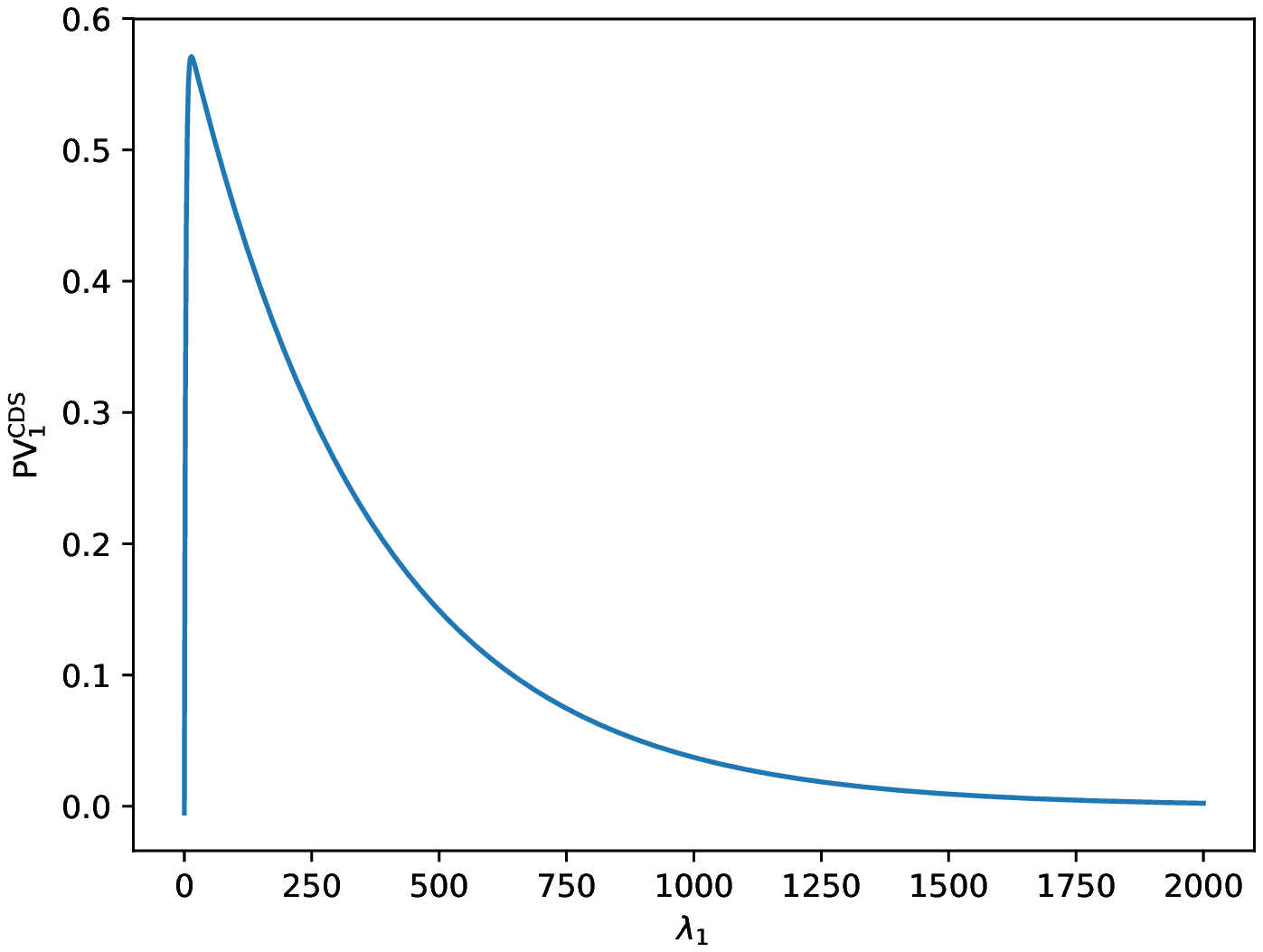}
     \caption{}\label{fig:PV1_extreme}
   \end{subfigure}
   \end{center}
   \caption{Present value of the first CDS used for the calibration example provided in Section \ref{sec:example} (6M maturity) with risk-neutral default probabilities implied via the minmaxvar distortion, as a function of $\lambda_1$ for the same range for $\lambda_1$ as in Figure \ref{fig:PV6}, panel (a), and when $\lambda_1$ diverges, panel (b). Notional assumed unitary.}
\end{figure}

Therefore, for parameters usually considered in practice, assuming that the price of the $i$\textsuperscript{th} CDS strictly increases when $\lambda_i$ increases (eventually within an interval that is large enough for practical applications in the case $i=1$) is a reasonable assumption, that we have used throughout the paper.

\section{Tables\label{sec:tables}}

\begin{table}[!htb]
\centering
\captionsetup{width=.8\textwidth}
\caption{Bid and ask CDS premia depicted in Figure \ref{fig:premia}, rounded to the basis-point digit. For comparison purposes, the mid CDS premia ($\UF^{\text{mid}}$), as well as the absolute value of the bid-ask CDS premium spreads with respect to mid CDS premia have been reported.\label{tab:premia}}
    \begin{tabular}{|c|c|c|c|c|}
    \hline
    \textbf{Tenor} & $\UF^{\text{bid}}_i$ & $\UF^{\text{ask}}_i$ & $\UF^{\text{mid}}_i$ & $|(\UF^{\text{ask}}_i-\UF^{\text{bid}}_i)/\UF^{\text{mid}}_i|$ \\
    \hline
    6M    & -0.0033 & -0.0026 & -0.0030 & 23.73\%\\
    1Y    & -0.0074 & -0.0068 &  -0.0071 & 8.45\%\\
    2Y    & -0.0149 & -0.0126 &  -0.0138 & 16.73\%\\
    3Y    & -0.0192 & -0.0169 &  -0.0181 & 12.74\%\\
    4Y    & -0.0221 & -0.0198 &  -0.0210 & 10.98\%\\
    5Y    & -0.0219 & -0.0198 &  -0.0209 & 10.07\%\\
    7Y    & -0.0162 & -0.0095 &  -0.0129 & 52.14\%\\
    10Y   & -0.0073 & 0.0047 &  -0.0013 & 932.08\%\\
    \hline
    \end{tabular}
\end{table}

\begin{table}[!htb]
\centering
\captionsetup{width=.9\textwidth}
    \caption{Calibrated parameters: hazard rates $\lambda$, panel (a), and distortion parameters $\gamma$, panel (b). Implied hazard rates and distortion parameters are depicted in Figures \ref{fig:lambda} and \ref{fig:gamma}, respectively.}
    \begin{subtable}{.5\linewidth}
      \centering        
    \begin{tabular}{|c|c|c|}
    \hline
    \textbf{Tenor} & $\lambda_i^{\text{minmaxvar}}$ & $\lambda_i^{\text{Wang}}$ \\
    \hline
    6M    & 0.001302 & 0.002235 \\
    1Y    & 0.003348 & 0.003353 \\
    2Y    & 0.004790 & 0.005615 \\
    3Y    & 0.009759 & 0.009705 \\
    4Y    & 0.011977 & 0.011906 \\
    5Y    & 0.017028 & 0.016865 \\
    7Y    & 0.022719 & 0.023514 \\
    10Y   & 0.023259 & 0.023671 \\
    \hline
    \end{tabular}%
     \caption{\label{tab:lambdas}}
    \end{subtable}%
    \begin{subtable}{.5\linewidth}
      \centering
        \begin{tabular}{|c|c|c|}
    \hline
      \textbf{Tenor} & $\gamma_i^{\text{minmaxvar}}$ & $\gamma_i^{\text{Wang}}$ \\
    \hline
    6M    & 0.159982 & 0.202429 \\
    1Y    & 0.055352 & 0.063821 \\
    2Y    & 0.076611 & 0.083899 \\
    3Y    & 0.040828 & 0.042300 \\
    4Y    & 0.027284 & 0.027351 \\
    5Y    & 0.017690 & 0.017241 \\
    7Y    & 0.033905 & 0.032033 \\
    10Y   & 0.041636 & 0.038361 \\
    \hline
    \end{tabular}%
    \caption{\label{tab:gammas}}
    \end{subtable} 
\end{table}

\section*{Acknowledgments}
\noindent We are grateful to Raoul Pietersz and Leo Kits for their valuable comments and suggestions. The opinions expressed in this paper are solely those of the authors and do not necessarily reflect those of their current and past employers.

\end{doublespace}

\clearpage

\bibliographystyle{apalike}
\bibliography{biblio}

\end{document}